\documentclass[11pt]{amsart}
\usepackage[margin=1in]{geometry}
\usepackage{amsfonts}
\usepackage{amssymb}
\usepackage[dvips]{graphics}
\usepackage{epsfig}
\usepackage{euscript}
\usepackage{color}

 \newtheorem{thm}{Theorem}[section]
 \newtheorem{cor}[thm]{Corollary}
 \newtheorem{lem}[thm]{Lemma}
 \newtheorem{prop}[thm]{Proposition}
 \theoremstyle{definition}
 
 \theoremstyle{remark}
 \newtheorem{rem}[thm]{Remark}
 \newtheorem*{ex}{Example}
 \numberwithin{equation}{section}
 
 \def\idty{{\mathchoice {\mathrm{1\mskip-4mu l}} {\mathrm{1\mskip-4mu l}} %
{\mathrm{1\mskip-4.5mu l}} {\mathrm{1\mskip-5mu l}}}}

\newcommand{\bR}{{\mathbb R}}

\newcommand{\cA}{{\mathcal A}}
\newcommand{\cB}{{\mathcal B}}
\newcommand{\cH}{{\mathcal H}}
\newcommand{\rd}{{\rm d}}

\newcommand{\cU}{{\mathcal U}}

\newcommand{\be}{\begin{equation}}
\newcommand{\ee}{\end{equation}}

\newcommand{\Rl}{\bR}

\begin{document}

\renewcommand{\thefootnote}{\fnsymbol{footnote}}
\title[On the Existence of the Dynamics]{On the Existence of the Dynamics for Anharmonic \\ Quantum Oscillator Systems}

\author[B. Nachtergaele]{Bruno Nachtergaele}
\address{Department of Mathematics\\
University of California, Davis\\
Davis, CA 95616, USA}
\email{bxn@math.ucdavis.edu}

\author[B. Schlein]{Benjamin Schlein}
\address{Centre for Mathematical Sciences \\
University of Cambridge \\
Cambridge, CB3 0WB, UK}
\email{b.schlein@dpmms.cam.ac.uk}

\author[R. Sims]{Robert Sims}
\address{Department of Mathematics\\
University of Arizona\\
Tucson, AZ 85721, USA}
\email{rsims@math.arizona.edu}

\author[S.Starr]{Shannon Starr}
\address{Department of Mathematics\\
University of Rochester\\
Rochester, NY 14627, USA}
\email{sstarr@math.rochester.edu}

\author[V. Zagrebnov]{Valentin Zagrebnov}
\address{
Universite de la M\'editerran\'ee (Aix-Marseille II)\\
Centre de Physique Th\'eorique- UMR 6207 CNRS, Luminy - Case 907\\
13288 Marseille, Cedex 09, France}
\email{zagrebnov@cpt.univ-mrs.fr}

\date{Version: \today }
\maketitle
\bigskip
\begin{abstract}
We construct a $W^*$-dynamical system describing the dynamics of a class of 
anharmonic quantum oscillator lattice systems in the thermodynamic limit.
Our approach is based on recently proved Lieb-Robinson bounds for such 
systems on finite lattices \cite{nachtergaele2009}.
\end{abstract}

\maketitle

\footnotetext[1]{Copyright \copyright\ 2009 by the authors. This
paper may be reproduced, in its entirety, for non-commercial
purposes.}

\section{Introduction}\label{sec:intro}

The dynamics of a finite quantum system, i.e., one with a finite number
of degrees of freedom described by a Hilbert space $\cH$, is given by the 
Schr\"odinger equation. The Hamiltonian $H$ is a densely 
defined self-adjoint operator on $\cH$, and for a vector $\psi(t)$ in
the domain of $H$  the state at time $t$ satisfies
\be
i\partial_t \psi(t) = H\psi(t) \, .
\label{se}\ee
For all initial conditions $\psi(0)\in\cH$, the unique solution is given by
$$
\psi(t)=e^{-itH}\psi(0), \mbox{ for all } t\in\Rl.
$$
Due to Stone's Theorem $e^{-itH}$ is a strongly continuous one-parameter group of
unitary operators on $\cH$, and the self-adjointness of $H$ is the necessary and sufficient
condition for the existence of a unique continuous solution for all times.

An alternative description of this dynamics is the so-called Heisenberg picture
in which the time evolution is defined on the algebra of observables instead
of the Hilbert space of states. The corresponding Heisenberg equation is
\be
\partial_t A(t)=i[H,A(t)]\, ,
\label{he}\ee
where, for each $t\in \Rl$, $A(t)\in\cB(\cH)$ is a bounded linear operator
on $\cH$. Its solutions are given by a one-parameter group of $*$-automorphisms,
$\tau_t$, of $\cB(\cH)$:
$$
A(t)=\tau_t(A(0)).
$$

For the description of physical systems we expect the Hamiltonian, $H$,
to have some additional properties. E.g., for finite systems such as atoms or
molecules, stability of the system requires that $H$ is bounded from below.
In this case, the infimum of the spectrum is expected to be an eigenvalue and
is called the ground state energy. When the model Hamiltonian, $H$, is describing 
bulk matter rather than finite systems, we expect some additional properties. 
E.g., the stability of matter requires that the ground state energy has a lower bound
proportional to $N$, where $N$ is the number of degree of freedom. Much progress
on this stability property has been made in the last several decades 
\cite{lieb_selecta_stability,lieb-seiringer}.
We also expect that the dynamics of local observables of bulk matter, or large
systems in general, depends only on the local environment. Mathematically this 
is best expressed by the existence of the dynamics in the thermodynamic limit,
i.e., in infinite volume. This is the question we address in this paper.

There are two settings that allow one to prove a rich set of important physical
properties of quantum dynamical systems, including infinite ones: the $C^*$ 
dynamical systems and the $W^*$ dynamical systems \cite{bratteli1987}. In both cases, the 
algebra of observables can be thought of a norm-closed $*$-subalgebra $\cA$ of some algebra
of the form $\cB(\cH)$, but in the case of the $W^*$-dynamical systems we additionally
require that the algebra is closed for the weak operator topology, which makes it a 
von Neumann algebra.  For a $C^*$-dynamical  system the group of automorphisms 
$\tau_t$ is assumed to be strongly continuous, i.e., for all $A\in\cA$, the map 
$t\mapsto \tau_t(A)$ is continuous in $t$ for the operator norm ($C^*-$norm) on $\cA$. 
In a $W^*$-dynamical system the continuity  is with respect to the weak topology.

In the case of lattice systems with a finite-dimensional Hilbert space
of states associated with each lattice sites, such as quantum spin-lattice systems 
and lattice fermions, it has been known for a long time that under rather
general conditions the dynamics can be described by a $C^*$ dynamical
system, including in the thermodynamic limit \cite{bratteli1997}.  When the Hilbert 
space at each site is infinite-dimensonal and the finite-system Hamiltonians are 
unbounded, this is no longer possible and the {\em weak continuity} becomes a natural 
assumption.

The class of systems we will primarily focus on here are lattices of quantum oscillators but
the underlying lattice structure is not essential for our method. Systems defined on 
suitable graphs, such as the systems considered in \cite{eisert2005,eisert2008} can
also be analyzed with the same methods. In a recent preprint \cite{amour2009}, it was shown 
that convergence of the dynamics in the thermodynamic limit can be obtained for a modified 
topology. Here, we follow a somewhat different approach. The main difference is that we study
the thermodynamic limit of anharmonic perturbations of an {\em infinite} harmonic lattice
system described by an explicit $W^*$-dynamical system. The more traditional way
is to first define the dynamics of anharmonic systems in finite volume (which can be done 
by standard means \cite{reed-simon}), and then to study the limit in which the volume 
tends to infinity.  This is what is done in \cite{amour2009}, but it appears that controlling 
the continuity of the limiting dynamics is more straightforward 
in our approach. In fact, we are able to show that the resulting dynamics for the class of 
anharmonic lattices we study is indeed weakly continuous, and we obtain a 
$W^*$-dynamical system for the infinite system. The $W^*$-dynamical setting is 
obtained by considering the GNS representation of a ground state or thermal 
equilibrium state of the harmonic system. The ground states and thermal states 
are quasi-free states in the sense of \cite{robinson1965}, or convex mixtures
of quasi-free states. In the ground state case the GNS representations are the well-known
Fock reprensentations. For the thermal states the GNS representations have been constructed 
by Araki and Woods \cite{araki1963}.

Common to both approaches, ours and the one of \cite{amour2009}, is the crucial role
played by an estimate of the speed of propagation of perturbations in the system, commonly
referred to as Lieb-Robinson bounds \cite{hast2006,lieb1972,nach12006,nach22006,nach2007}.
Briefly, if $A$ and $B$ are two 
observables of a spatially extended system, localized in regions $X$ and $Y$ of our graph,
respectively, and $\tau_t$ denotes the time evolution of the system then, a Lieb-Robinson
bound is an estimate of the form
$$
\Vert [\tau_t (A), B]\Vert\leq C e^{-a(d(X,Y)-v\vert t\vert)}\, ,
$$
where $C, a$, and $v$ are positive constants and $d(X,Y)$ denotes the distance
between $X$ and $Y$. Lieb-Robinson bounds for anharmonic 
lattice systems were recently proved in \cite{nachtergaele2009}, and this work builds on the
results obtained there. Our results are mainly limited to short-range interactions that are either 
bounded or unbounded perturbations of the harmonic interaction (linear springs).

To conclude the introduction, let us mention that the same questions, the existence of the
dynamics for infinite oscillator lattices, can and has been asked for classical systems. Two
classic papers are \cite{lanford1977,marchioro1978}. Many properties of this classical infinite
volume harmonic dynamics have been studied in detail e.g. \cite{Spohn77, vanhem} and 
some recent progress on locality estimates for anharmonic systems is reported in 
\cite{butta2007,raz2009}.

The paper is organized as follows. We begin with a section discussing bounded interactions.
In this case, the existence of the dynamics follows by mimicking the proof valid in the 
context of quantum spins systems. Section 3 describes the infinite volume harmonic dynamics on
general graphs. It is motivated by an explicit example on $\mathbb{Z}^d$. Next, in Section 4, we 
discuss finite volume perturbations of the infinite volume harmonic dynamics and prove that
such systems satisfy a Lieb-Robinson bound. In Section 5 we demonstrate that the existence
of the dynamics and its continuity follow from the Lieb-Robinson estimates established in 
the previous section.

%
%
%
%

\section{Bounded Interactions} \label{sec:bdints}

The goal of this section is to prove the existence of the dynamics for oscillator systems
with bounded interactions. Since oscillator systems with bounded interactions can be treated 
as a special case of more general models with bounded interactions, we will use 
a slightly more general setup in this section, which we now introduce.

We will denote by $\Gamma$ the underlying structure on which our models will be defined.
Here $\Gamma$ will be an arbitrary set of sites equipped with a metric $d$.
For $\Gamma$ with countably infinite cardinality, we will need to assume that there exists a
non-increasing function $F: [0, \infty) \to (0, \infty)$ for which:

\noindent i) $F$ is uniformly integrable over $\Gamma$, i.e.,
\begin{equation} \label{eq:fint}
\| \, F \, \| \, := \, \sup_{x \in \Gamma} \sum_{y \in \Gamma}
F(d(x,y)) \, < \, \infty,
\end{equation}

\noindent and

\vspace{.3cm}

\noindent ii) $F$ satisfies
\begin{equation} \label{eq:intlat}
C \, := \, \sup_{x,y \in \Gamma} \sum_{z \in \Gamma}
\frac{F \left( d(x,z) \right) \, F \left( d(z,y)
\right)}{F \left( d(x,y) \right)} \, < \, \infty.
\end{equation}

Given such a set $\Gamma$ and a function $F$, by the triangle inequality,
for any $a \geq 0$ the function
\begin{equation*}
F_a(x) = e^{-ax} \, F(x),
\end{equation*}
also satisfies i) and ii) above with $\| F_a \| \leq \| F \|$ and $C_a \leq C$.

In typical examples, one has that $\Gamma  \subset \mathbb{Z}^{d}$ for
some integer $d \geq 1$, and the metric is
just given by $d(x,y) = |x -  y|=\sum_{j=1}^{d} |x_j - y_j|$.
In this case, the function $F$ can be
chosen as $F(|x|) = (1 + |x|)^{- d - \epsilon}$ for any $\epsilon >0$.

To each $x \in \Gamma$, we will associate a Hilbert space $\mathcal{H}_x$.
In many relevant systems, one considers
$\mathcal{H}_x = L^2( \mathbb{R}, \rd q_x)$, but this is not essential.
With any finite subset $\Lambda \subset \Gamma$,
the Hilbert space of states over $\Lambda$ is given by
\begin{equation*}
\mathcal{H}_{\Lambda} \, = \, \bigotimes_{x \in \Lambda} \mathcal{H}_x,
\end{equation*}
and the local algebra of observables over $\Lambda$ is then defined to be
\[
\mathcal{A}_{\Lambda} = \bigotimes_{x \in \Lambda} \cB (\cH_x),
\]
where $\cB (\cH_x)$ denotes the algebra of bounded linear operators on $\cH_x$.

If $\Lambda_1 \subset \Lambda_2$, then there is a natural way of identifying
$\mathcal{A}_{\Lambda_1} \subset \mathcal{A}_{\Lambda_2}$, and we may thereby
define the algebra of quasi-local observables by the inductive limit
\begin{equation*}
\mathcal{A}_{\Gamma} \, = \, \bigcup_{\Lambda \subset \Gamma} \mathcal{A}_{\Lambda},
\end{equation*}
where the union is over all finite subsets $\Lambda \subset \Gamma$; see
\cite{bratteli1987,bratteli1997} for a discussion of these issues in general.

The result discussed in this section corresponds to bounded perturbations of
local self-adjoint Hamiltonians. We fix a collection of on-site local operators
$H^{\rm loc} = \{ H_x \}_{x \in \Gamma}$ where each $H_x$ is a self-adjoint
operator over $\mathcal{H}_x$. In addition, we will consider a general class of bounded perturbations.
These are defined in terms of an interaction $\Phi$, which is a map from the
set of subsets of $\Gamma$ to $\mathcal{A}_{\Gamma}$ with the property that
for each finite set $X \subset \Gamma$, $\Phi(X) \in \mathcal{A}_X$ and
$\Phi(X) ^*= \Phi(X)$. As with the Lieb-Robinson bound proven in \cite{nachtergaele2009}, 
we will need a growth condition on the set of interactions $\Phi$ for which we can prove
the existence of the dynamics in the thermodynamic limit.
This condition is expressed in terms of the following norm. 
For any $a \geq 0$, denote by $\mathcal{B}_a(\Gamma)$ the set of interactions
for which
\begin{equation} \label{eq:defphia}
\| \Phi \|_a \, := \, \sup_{x,y \in \Gamma}  \frac{1}{F_a (d(x,y))} \,
\sum_{X \ni x,y} \| \Phi(X) \| \, < \, \infty.
\end{equation}

Now, for a fixed sequence of local Hamiltonians $H^{\rm loc} = \{H_x \}_{x\in\Gamma}$, as
described above, an interaction $\Phi \in \mathcal{B}_a(\Gamma)$, and a finite subset
$\Lambda \subset \Gamma$, we will consider self-adjoint Hamiltonians of the form
\begin{equation} \label{eq:localham}
H_{\Lambda} \, = \, H^{\rm loc}_{\Lambda} \, + \, H^{\Phi}_{\Lambda} \, = \, \sum_{x \in \Lambda} H_x \, + \, \sum_{X \subset \Lambda} \Phi(X),
\end{equation}
acting on $\mathcal{H}_{\Lambda}$ (with domain given by $\bigotimes_{x \in \Lambda} D(H_x)$ where $D(H_x) \subset \cH_x$ denotes the domain of $H_x$). As these operators are self-adjoint, they generate a dynamics, or time evolution, $\{ \tau_t^{\Lambda} \}$,
which is the one parameter group of automorphisms defined by
\begin{equation*}
\tau_t^{\Lambda}(A) \, = \, e^{it H_{\Lambda}} \, A \, e^{-itH_{\Lambda}} \quad \mbox{for any} \quad A \in \mathcal{A}_{\Lambda}.
\end{equation*}

\begin{thm}
Under the conditions stated above, for all $t \in \bR$, $A \in \mathcal{A}_{\Gamma}$, 
the norm limit 
\begin{equation}\label{eq:claim} 
\lim_{\Lambda \to \Gamma} \, \tau_t^{\Lambda} (A) = \tau_t(A)
\end{equation} exists in the sense of non-decreasing exhaustive sequences of finite
volumes $\Lambda$ and defines a group of $*-$automorphisms $\tau_t$ on the completion of
$\mathcal{A}_\Gamma$. The convergence is uniform for $t$ in a compact set.
\end{thm}

\begin{proof} 
Let $\Lambda \subset \Gamma$ be a finite set. Consider the unitary propagator
\begin{equation} \label{eq:intuni}
 \cU_{\Lambda} (t,s) = e^{i t H_{\Lambda}^{\text{loc}} } \, e^{-i (t-s) H_{\Lambda}} \, e^{-is H_{\Lambda}^{\text{loc}}} 
\end{equation}
and its associated {\it interaction-picture} evolution defined by
\begin{equation} \label{eq:intpic}
\tau^{\Lambda}_{t, \text{int}} (A) = \cU_{\Lambda} (0,t) \, A \; \cU_{\Lambda} (t,0) \quad \mbox{for all } A \in \mathcal{A}_{\Gamma} \, .
\end{equation} 
Clearly, $\mathcal{U}_{\Lambda}(t,t) = \idty$ for all $t \in \mathbb{R}$, and it is also easy to check 
that 
\begin{equation*}
i \frac{\rd}{\rd t} \, \cU_{\Lambda} (t,s) =  H_{\Lambda}^{\text{int}} (t) \, \cU_{\Lambda} (t,s)  \quad \mbox{and} \quad
- i \frac{\rd}{\rd s} \, \cU_{\Lambda} (t,s) =  \cU_{\Lambda} (t,s) \, H_{\Lambda}^{\text{int}} (s) 
\end{equation*}
with the time-dependent generator 
\begin{equation} \label{eq:gen}
H^{\text{int}}_{\Lambda} (t) = e^{i H_{\Lambda}^{\text{loc}} t} H_{\Lambda}^{\Phi} e^{-i H^{\text{loc}}_{\Lambda} t} = \sum_{Z \subset \Lambda} e^{i H_{\Lambda}^{\text{loc}} t} \, \Phi (Z)  \, e^{-i H^{\text{loc}}_{\Lambda} t} \, . 
\end{equation}

Fix $T>0$ and $X \subset \Gamma$ finite. For any $A \in \mathcal{A}_X$, we will show that 
for any non-decreasing, exhausting sequence $\{ \Lambda_n \}$ of $\Gamma$, the sequence
$\{ \tau_{t, \text{int}}^{\Lambda_n}(A) \}$ is Cauchy in norm, uniformly for $t \in [-T,T]$.
Moreover, the bounds establishing the Cauchy property depend on $A$ only through $X$ and $\|A\|$.
Since
\begin{equation*}
\tau_t^{\Lambda} (A) = \tau_{t,\text{int}}^{\Lambda} \left(e^{itH_{\Lambda}^{\text{loc}}} \, A \,  e^{-it H_{\Lambda}^{\text{loc}}} \right) = 
\tau_{t,\text{int}}^{\Lambda} \left(e^{it \sum_{x \in X}H_x} \, A \, e^{-i t \sum_{x \in X} H_x} \right) \, ,
\end{equation*}
an analogous statement then immediately follows for $\{ \tau_t^{\Lambda_n}(A) \}$, since they
are all also localized in $X$ and have the same norm as $\|A\|$.

Take $n \leq m$ with $X \subset \Lambda_n \subset \Lambda_m$ and calculate
\begin{equation} \label{eq:diff}
\tau_{t,\text{int}}^{\Lambda_m} (A) - \tau_{t,\text{int}}^{\Lambda_n} (A) = \int_0^t \frac{\rd}{\rd s} \left\{ \cU_{\Lambda_m} (0,s) \, \cU_{\Lambda_n} (s,t) \, A \, \cU_{\Lambda_n} (t,s) \, \cU_{\Lambda_m} (s,0) \right\} \, ds \, .
\end{equation}
A short calculation shows that
\begin{equation}
\begin{split}
\frac{\rd}{\rd s} \cU_{\Lambda_m} (0,s) & \, \cU_{\Lambda_n} (s,t) \, A \, \cU_{\Lambda_n} (t,s) \, \cU_{\Lambda_m} (s,0) \\
& = \, i \mathcal{U}_{\Lambda_m}(0,s) \left[ \left( H^{\text{int}}_{\Lambda_m}(s) - H^{\text{int}}_{\Lambda_n}(s) \right), \mathcal{U}_{\Lambda_n}(s,t) \, A \, \mathcal{U}_{\Lambda_n}(t,s) \right] \mathcal{U}_{\Lambda_m}(s,0) \\
& = \, i \mathcal{U}_{\Lambda_m}(0,s) e^{is H_{\Lambda_n}^{\text{loc}}} \left[ \tilde{B}(s), \tau_{s-t}^{\Lambda_n} \left( \tilde{A}(t) \right) \right] e^{-is H_{\Lambda_n}^{\text{loc}}} \mathcal{U}_{\Lambda_m}(s,0) \, ,
\end{split}
\end{equation}
where
\begin{equation} \label{eq:tat}
\tilde{A}(t) = e^{-it H_{\Lambda_n}^{\text{loc}}} A \, e^{it H_{\Lambda_n}^{\text{loc}}} = e^{-it H_{X}^{\text{loc}}} A \, e^{it H_{X}^{\text{loc}}} 
\end{equation}
and
\begin{eqnarray} \label{eq:tbs}
\tilde{B}(s) & = & e^{-is H_{\Lambda_n}^{\text{loc}}}\left( H^{\text{int}}_{\Lambda_m}(s) - H^{\text{int}}_{\Lambda_n}(s) \right) e^{is H_{\Lambda_n}^{\text{loc}}} \nonumber \\
& = & \sum_{Z \subset \Lambda_m} e^{is H_{\Lambda_m \setminus \Lambda_n}^{\text{loc}}} \Phi(Z)  e^{-is H_{\Lambda_m \setminus \Lambda_n}^{\text{loc}}} - \sum_{Z \subset \Lambda_n} \Phi(Z) \nonumber \\
& = & \sum_{\stackrel{Z \subset \Lambda_m:}{ Z \cap \Lambda_m \setminus \Lambda_n \neq \emptyset}} e^{is H_{\Lambda_m \setminus \Lambda_n}^{\text{loc}}} \Phi(Z)  e^{-is H_{\Lambda_m \setminus \Lambda_n}^{\text{loc}}} 
\end{eqnarray}
Combining the results of (\ref{eq:diff}) -(\ref{eq:tbs}), and using unitarity, we find that
\begin{equation}
\left\| \tau_{t,\text{int}}^{\Lambda_m} (A) - \tau_{t,\text{int}}^{\Lambda_n} (A)  \right\| \leq \int_0^t \left\| \left[ \tau_{s-t}^{\Lambda_n} \left( \tilde{A}(t) \right), \tilde{B}(s)  \right] \right\| \, ds \,
\end{equation}
and by the Lieb-Robinson bound proven in \cite{nachtergaele2009}, it is clear that
\begin{eqnarray}
&&\left\| \left[ \tau_{s-t}^{\Lambda_n} \left( \tilde{A}(t) \right), \tilde{B}(s)  \right] 
\right\|\\
& \leq &  \sum_{\stackrel{Z \subset \Lambda_m:}{ Z \cap \Lambda_m \setminus 
\Lambda_n \neq \emptyset}} 
\left\| \left[ \tau_{s-t}^{\Lambda_n} \left( \tilde{A}(t) \right),  
e^{is H_{\Lambda_m \setminus \Lambda_n}^{\text{loc}}} \Phi(Z)  
e^{-is H_{\Lambda_m \setminus \Lambda_n}^{\text{loc}}}  \right] \right\| \nonumber \\
& \leq & \frac{2 \| A \|}{C_a} \left( e^{2 \| \Phi \|_a C_a |t-s|} - 1 \right)  \sum_{y \in \Lambda_m \setminus \Lambda_n} \sum_{\stackrel{Z \subset \Lambda_m:}{ y \in Z }}  \| \Phi(Z) \| \sum_{x \in X} \sum_{z \in Z} F_a( d(x,z)) \nonumber \\ & \leq &  \frac{2 \| A \|}{C_a} \left( e^{2 \| \Phi \|_a C_a |t-s|} - 1 \right)  \sum_{y \in \Lambda_m \setminus \Lambda_n} \sum_{z \in \Lambda_m} \sum_{\stackrel{Z \subset \Lambda_m:}{ y, z \in Z }}  \| \Phi(Z) \| \sum_{x \in X}  F_a( d(x,z)) \nonumber \\
& \leq & \frac{2 \| A \| \| \Phi \|_a}{C_a} \left( e^{2 \| \Phi \|_a C_a |t-s|} - 1 \right)  \sum_{y \in \Lambda_m \setminus \Lambda_n}  \sum_{x \in X}  \sum_{z \in \Lambda_m}  F_a( d(x,z)) F_a(d(z,y)) \nonumber \\
& \leq & 2 \| A \| \| \Phi \|_a \left( e^{2 \| \Phi \|_a C_a |t-s|} - 1 \right)  \sum_{y \in \Lambda_m \setminus \Lambda_n}  \sum_{x \in X} F_a( d(x,y)) \, .\nonumber
\end{eqnarray}
With the estimate above and the properties of the function $F_a$, it is clear that
\begin{equation}
\sup_{t \in [-T,T]} \left\| \tau_{t,\text{int}}^{\Lambda_m} (A) - \tau_{t,\text{int}}^{\Lambda_n} (A)  \right\| \to 0 \quad \mbox{ as } n, m \to \infty \, ,
\end{equation}
and the rate of convergence only depends on the norm $\| A \|$ and the set $X$ where $A$ is
supported. This proves the claim.
\end{proof}

If all
local Hamiltonians $H_x$ are bounded, $\{\tau_t\}$ is strongly continuous. 
If the $H_x$ are allowed to be densely defined unbounded self-adjoint operators, 
we only have weak continuity and the dynamics is more naturally defined on a
von Neumann algebra. This can be done when we have a suffiently nice invariant
state for the model with only the on-site Hamiltonians. E.g., suppose that
for each $x\in \Gamma$, we have a normalized eigenvector $\phi_x$ of $H_x$.
Then, for all $A\in\mathcal{A}_\Lambda$, for any finite $\Lambda\subset \Gamma$,
define
\begin{equation}
\rho(A)=\langle \bigotimes_{x\in\Lambda}\phi_x, A \bigotimes_{x\in\Lambda}\phi_x\rangle\, .
\end{equation}
$\rho$ can be regarded as a state of the infinite system defined on the norm completion 
of $\mathcal{A}_\Gamma$. The GNS Hilbert space $\mathcal{H}_\rho$
of $\rho$ can be constructed as the closure of 
$\mathcal{A}_\Gamma \bigotimes_{x\in\Gamma}\phi_x$. 
Let $\psi\in \mathcal{A}_\Gamma \bigotimes_{x\in\Gamma}\phi_x$.
Then 
\begin{equation}
\begin{split}
\left\| \left( \tau_t(A) - \tau_{t_0}(A) \right)  \psi \right\| \leq & \left\| 
\left( \tau_t(A) - \tau_t^{(\Lambda_n)}(A) \right)  \psi \right\| \\
+ & \left\| \left( \tau_t^{(\Lambda_n)}(A) - \tau_{t_0}^{(\Lambda_n)}(A) \right)  
\psi \right\| +
\left\| \left( \tau_{t_0}^{(\Lambda_n)}(A) - \tau_{t_0}(A) \right)  \psi \right\| \, ,
\end{split}
\end{equation}
For sufficiently large $\Lambda_n$, the $\lim_{t\to t_0}$ of middle term vanishes 
by Stone's theorem. The two other terms are handled by \ref{eq:claim}. It is clear
how to extend the continuity to $\psi\in\mathcal{H}_\rho$.
 
We will discuss this type of situation in more detail in the next three sections
where we consider models that include quadratic (unbounded) interactions as well.

%
%
%
%

\section{The Harmonic Lattice}\label{sec:harm}

As noted in the introduction, we will consider anharmonic perturbations of infinite
harmonic lattices. In this section we discuss the properties
of the harmonic systems that we need to assume in general in order
to study the perturbations in the thermodynamic limit. We will also
show in detail that a standard harmonic lattice model posesses all the
required properties.

\subsection{The CCR algebra of observables}

We begin by introducing the CCR algebra on which the harmonic dynamics will be
defined. Following \cite{manuceau1973}, one can define the CCR algebra over any 
real linear space
$\mathcal{D}$ equipped with a non-degenerate, symplectic bilinear form $\sigma$, i.e.
$\sigma : \mathcal{D} \times \mathcal{D} \to \mathbb{R}$ with the property that
if $\sigma(f,g) = 0$ for all $f \in \mathcal{D}$, then $g = 0$, and 
\begin{equation} \label{eq:symp}
\sigma(f,g) = - \sigma(g,f) \quad \mbox{for all } f, g \in \mathcal{D} .
\end{equation}
In typical examples, $\mathcal{D}$ will be a complex inner product space associated 
with $\Gamma$, e.g.
$\mathcal{D} = \ell^2( \Gamma)$ or a subspace thereof such as $\mathcal{D} = \ell^1( \Gamma)$, 
or $\ell^2(\Gamma_0)$, with $\Gamma_0\subset\Gamma$, and
\begin{equation}
\sigma(f,g) = \mbox{Im} \left[ \langle f, g \rangle \right] \, .
\end{equation}
The Weyl operators over $\mathcal{D}$ are defined by associating non-zero elements
$W(f)$ to each $f \in \mathcal{D}$ which satisfy
\begin{equation} \label{eq:invo}
W(f)^* = W(-f) \quad \mbox{for each } f \in \mathcal{D} \, ,
\end{equation}
and
\begin{equation} \label{eq:weylrel}
W(f) W(g) = e^{-i \sigma(f,g)/2} W(f+g) \quad \mbox{for all } f, g \in \mathcal{D} \, .
\end{equation}
It is well-known that there is a unique, up to $*$-isomorphism, $C^*$-algebra generated
by these Weyl operators with the property that $W(0) = \idty$, $W(f)$ is unitary
for all $f \in \mathcal{D}$, and $\| W(f) - \idty \| = 2$ for all $ f \in \mathcal{D} \setminus \{0 \}$, 
see e.g. Theorem 5.2.8 \cite{bratteli1997}. This algebra, commonly known as the CCR algebra, or Weyl algebra, over
$\mathcal{D}$, we will denote by $\mathcal{W} = \mathcal{W}( \mathcal{D})$.

\subsection{Quasi-free dynamics}
The anharmonic dynamics we study in this paper will be defined as
perturbations of harmonic, technically {\em quasi-free}, dynamics.
A quasi-free dynamics on $\mathcal{W}(\mathcal{D})$ is a one-parameter
group of *-automorphisms $\tau_t$ of the form
\begin{equation}
\tau_t(W(f))=W(T_t f), \quad f\in \mathcal{D}
\end{equation}
where $T_t:\mathcal{D}\to\mathcal{D}$ is a group of real-linear, symplectic
transformations, i.e., 
\begin{equation} \label{eq:sympT}
\sigma(T_t f, T_t g)= \sigma(f,g)\, .
\end{equation}
As $\| W(f) - W(g) \| = 2$ for all $ f\neq g \in \mathcal{D}$, one should not
expect $\tau_t$ to be strongly continuous; only a weaker form of
continuity is present. This means that $\tau_t$ does {\em not}
define a $C^*$-dynamical system on $\mathcal{W}$, and thus we 
look for a $W^*$-dynamical setting in which the weaker form of
continuity is naturally expressed. 

In the present context, it suffices to regard a {\it $W^*$-dynamical system} as a pair $\{ \mathcal{M}, \alpha_t \}$
where $\mathcal{M}$ is a von Neumann algebra and $\alpha_t$ is a weakly continuous, 
one parameter group of $*$-automorphisms of $\mathcal{M}$. For the harmonic systems we are 
considering, a specific  $W^*$-dynamical system arises as follows. Let $\rho$ be a 
state on $\mathcal{W}$ and denote by $\left( \mathcal{H}_{\rho}, \pi_{\rho}, 
\Omega_{\rho} \right)$ the corresponding GNS representation. We will assume
that $\rho$ is both regular and $\tau_t$-invariant. Recall that $\rho$ is 
regular if and only if $t \mapsto \rho( W(tf) )$ is continuous 
for all $f \in \mathcal{D}$, and $\tau_t$-invariance means
\begin{equation}
\rho( \tau_t(A) ) = \rho(A) \quad \mbox{for all } A \in \mathcal{W}.
\end{equation}
For the von Neumann 
algebra $\mathcal{M}$, take the weak-closure of $\pi_{\rho}( \mathcal{W})$ in $\mathcal{L}
( \mathcal{H}_{\rho})$ and let $\alpha_t$ be the weakly continuous, 
one parameter group of $*$-automorphisms of $\mathcal{M}$ obtained by lifting $\tau_t$ to 
$\mathcal{M}$. The latter step is possible since $\rho$ is $\tau_t$-invariant, see e.g. Corollary 2.3.17 \cite{bratteli1987}.

\subsection{Lieb-Robinson bounds for harmonic lattices}
To prove the existence of the dynamics for anharmonic models, we use that the unperturbed harmonic
system satisfies a Lieb-Robinson bound. Such an estimate depends directly on
properties of $\sigma$ and
$T_t$. In fact, it is easy to calculate that
\begin{eqnarray}
\left[ \tau_t(W(f)), W(g) \right]  & = & \left\{ W( T_t f) - W(g) W( T_tf) W(-g) \right\} W(g)  \nonumber \\
& = & \left\{ 1 - e^{i \sigma( T_tf, g)} \right\} W( T_tf) W(g) \, ,
\end{eqnarray}
using the Weyl relations (\ref{eq:weylrel}).  For the examples we consider below, one can prove that
for every $a>0$, there exists positive numbers $c_a$ and $v_a$ for which 
\begin{equation} \label{eq:dynestex}
\left| \sigma( T_t f, g ) \right| \leq c_a e^{ v_a |t|} \sum_{x, y \in \mathbb{Z}^d} |f(x)| \, |g(y)| \frac{e^{-a|x-y|}}{(1+|x-y|)^{d+1}}
\end{equation}
holds for all $t \in \mathbb{R}$ and all $f, g \in \ell^2( \mathbb{Z}^d)$. In general, we will assume that
the harmonic dynamics satisfies an estimate of this type. Namely, we suppose that there exists a number
$a_0 >0$ for which given $0<a \leq a_0$, there are numbers $c_a$ and $v_a$ for which 
\begin{equation} \label{eq:dynest2}
\left| 1 - e^{i\sigma( T_t f, g )} \right| \leq c_a e^{ v_a |t|} \sum_{x, y \in \Gamma} |f(x)| \, |g(y)| F_a \left( d(x,y) \right) 
\end{equation}
holds for all $t \in \mathbb{R}$ and all $f, g \in \ell^2( \Gamma)$. Here we describe the spatial decay in $\Gamma$
through the functions $F_a$ as introduced in Section~\ref{sec:bdints}.  Since the Weyl operators are unitary, the norm estimate
\begin{equation}  \label{eq:freest}
\left\| \left[ \tau_t(W(f)), W(g) \right]  \right\| \leq  c_a e^{v_a |t| } \sum_{x, y} |f(x)| \, |g(y)| \, F_a \left( d(x,y) \right) \, ,
\end{equation}
readily follows.

\subsection{An important example}

Using the example given below, we illustrate the general discussion above in terms of 
a standard harmonic model defined over $\Gamma = \mathbb{Z}^d$.
We begin with a description of some well known calculations that are valid
for these models when restricted to a finite volume. This analysis motivates
the definition of the harmonic dynamics in the infinite volume. We then 
demonstrate that this infinite volume dynamics satisfies a Lieb-Robinson bound. 
By representing this dynamics in a suitable state, the relevant weak-continuity 
is readily verified. Interestingly, our analysis also applies to the massless case
of $\omega =0$, see below, and we discuss this briefly. We end this subsection
with some final comments.

\subsubsection{Finite volume analysis} We consider a system of coupled harmonic oscillators restricted to a finite volume. 
Specifically on cubic subsets $\Lambda_L \, = \, \left( -L, L \right]^d  \subset \mathbb{Z}^d$,
we analyze Hamiltonians of the form
\begin{equation} \label{eq:harham}
H_L^{h} \, = \,  \sum_{ x \in \Lambda_L} p_{ x }^2 \, +\, \omega^2 \, q_{ x}^2 \, + \,
\sum_{j = 1}^{d}  \lambda_j \, (q_{ x } - q_{ x + e_j})^2
\end{equation}
acting in the Hilbert space
\begin{equation} \label{eq:hspace}
\mathcal{H}_{\Lambda_L} = \bigotimes_{x \in \Lambda_L} L^2(\mathbb{R}, dq_x).
\end{equation}
Here the quantities $p_x$ and $q_x$, which appear in (\ref{eq:harham}) above, are the 
single site momentum and position operators regarded as operators on the full Hilbert space $\mathcal{H}_{\Lambda_L}$
by setting 
\begin{equation} \label{eq:pandq}
p_x = \idty \otimes \cdots \otimes \idty
\otimes -i \frac{d}{dq} \otimes \idty \cdots \otimes \idty \quad
\mbox{ and } \quad q_x = \idty \otimes \cdots \otimes \idty \otimes q \otimes \idty
\cdots \otimes \idty,
\end{equation}
i.e., these operators act non-trivially only in the $x$-th factor of $\mathcal{H}_{\Lambda_L}$. These operators satisfy the canonical commutation relations
\begin{equation} \label{eq:comm}
[p_x, p_y] \, = \, [q_x, q_y] \, = \, 0 \quad \mbox{ and } \quad
[q_x, p_y] \, = \, i \delta_{x,y},
\end{equation}
valid for all $x, y \in \Lambda_L$.  In addition,  $\{ e_j \}_{j=1}^{d}$ are the canonical basis vectors in $\mathbb{Z}^{d}$, the numbers
$\lambda_j \geq 0$ and $\omega \geq 0$ are the parameters of the system, and the Hamiltonian
is assumed to have periodic boundary conditions, in the sense that $q_{x+e_j} = q_{x-(2L-1)e_j}$ 
if $x \in \Lambda_L$ but $x+ e_j \not\in \Lambda_L$. It is well-known that Hamiltonians 
of this form can be diagonalized in Fourier space. We review this quickly to establish some notation and
refer the interested reader to \cite{nachtergaele2009} for more details.
 
Introducing the operators
\begin{equation} \label{eq:Q+Pk}
Q_k \, = \, \frac{1}{ \sqrt{ | \Lambda_L |}} \sum_{x \in \Lambda_L} e^{- i k \cdot x} q_x \quad \mbox{and} \quad
P_k \, = \, \frac{1}{ \sqrt{ | \Lambda_L |}} \sum_{x \in \Lambda_L} e^{- i k \cdot x} p_x \, ,
\end{equation}
defined for each $k \in \Lambda_L^* \, = \, \left\{ \, \frac{x \pi}{L} \, : \, x \in \Lambda_L \, \right\} $,
and setting 
\begin{equation} \label{eq:defgamma}
\gamma(k) \, = \,  \sqrt{ \omega^2 \, + \, 4 \sum_{j=1}^{d} \lambda_j \, \sin^2(k_j/2) },
\end{equation}
one finds that 
\begin{equation} \label{eq:diagham}
H_L^h \, = \, \sum_{k \in \Lambda_L^*} \, \gamma(k) \, \left( \, 2 \, b_k^*\,  b_k \, + \, 1 \, \right) \, 
\end{equation}
where the operators $b_k$ and $b_k^*$ satisfy
\begin{equation} \label{eq:beqns}
b_k \, = \, \frac{1}{ \sqrt{2 \gamma(k)}} \, P_k - i \sqrt{ \frac{\gamma(k)}{2}} \, Q_k \quad {\rm and} \quad
b_k^* \, = \, \frac{1}{ \sqrt{2 \gamma(k)}} \, P_{-k} + i \sqrt{ \frac{\gamma(k)}{2}} \, Q_{-k} \, .
\end{equation}
In this sense, we regard the Hamiltonian $H_L^h$ as diagonalizable.

Using the above diagonalization, one can determine the action of the dynamics corresponding to 
$H_L^h$ on the Weyl algebra $\mathcal{W}( \ell^2( \Lambda_L) )$. In fact, by setting
\begin{equation}
W(f) = \mbox{exp} \left[ i \sum_{x \in \Lambda_L} \mbox{Re}[f(x)] q_x + \mbox{Im}[f(x)] p_x \right] \, ,
\end{equation}
for each $f \in \ell^2( \Lambda_L)$, it is easy to verify that (\ref{eq:invo}) and (\ref{eq:weylrel}) hold with
$\sigma(f,g) = \mbox{Im}[ \langle f, g \rangle]$.
It is convenient to express these Weyl operators in terms of annihilation and creation operators, i.e.,
\begin{equation} \label{eq:anncre}
a_x \, = \, \frac{1}{\sqrt{2}} \left( q_x \, + \, i p_x \right) \quad \mbox{and} \quad a^*_x \, = \, \frac{1}{\sqrt{2}} \left( q_x \, - \, i p_x \right),
\end{equation}
which satisfy 
\begin{equation} \label{eq:comrela}
[a_x, a_y] = [a_x^*, a_y^*] = 0 \quad \mbox{and} \quad [a_x, a_y^*] = \delta_{x,y} \quad \mbox{for all } x,y \in \Lambda_L \, .
\end{equation}
One finds that
\begin{equation} \label{eq:weyla}
W(f)  =  \mbox{exp} \left[ \frac{i}{\sqrt{2}} \left( a(f) \, + \, a^*(f) \right) \right] \, ,
\end{equation}
where, for each $f \in \ell^2(\Lambda_L)$, we have set
\begin{equation} \label{eq:defafa*f}
a(f) \, = \, \sum_{x \in \Lambda_L} \overline{f(x)} \, a_x,  \quad a^*(f) \, = \, \sum_{x \in \Lambda_L} f(x) \, a_x^*\, .
\end{equation}

Now, the dynamics corresponding to $H_L^h$, which we denote by $\tau_t^{L}$, is trivial with 
respect to the diagonalizing variables, i.e., 
\begin{equation}
\tau_t^{L}(b_k) = e^{-2i \gamma(k) t} b_k \quad \mbox{and} \quad \tau_t^{L}(b_k^*) = e^{2i \gamma(k) t} b_k^* \, , 
\end{equation}
where $b_k$ and $b_k^*$ are as defined in (\ref{eq:beqns}). 
Hence, if we further introduce
\begin{equation} \label{eq:defb}
b_x = \frac{1}{ \sqrt{| \Lambda_L|}} \sum_{k \in \Lambda_L^*} e^{ikx} b_k \quad \mbox{and} \quad 
b_x^* = \frac{1}{ \sqrt{| \Lambda_L|}} \sum_{k \in \Lambda_L^*} e^{ikx} b_k^*, 
\end{equation}
for each $x \in \Lambda_L$ and, analogously to (\ref{eq:defafa*f}), define
\begin{equation} \label{eq:defbfb*f}
b(f) \, = \, \sum_{x \in \Lambda_L} \overline{f(x)} \, b_x,  \quad b^*(f) \, = \, \sum_{x \in \Lambda_L} f(x) \, b_x^*,
\end{equation}
for each $f \in \ell^2( \Lambda_L)$, then one has that
\begin{equation} \label{eq:taub}
\tau_t^L \left( b(f) \right) = b \left( [\mathcal{F}^{-1} M_t \mathcal{F}] f \right) \, , 
\end{equation}
where $\mathcal{F}$ is the unitary Fourier transform on $\ell^2(\Lambda_L)$ and $M_t$ is the 
operator of multiplication by $e^{2i \gamma(k) t}$ in Fourier space with $\gamma(k)$ as in (\ref{eq:defgamma}). We need only determine
the relation between the $a$'s and the $b$'s.

A short calculation shows that there exists a linear mapping $U: \ell^2( \Lambda_L) \to \ell^2( \Lambda_L)$
and an anti-linear mapping $V: \ell^2( \Lambda_L) \to \ell^2( \Lambda_L)$ for which
\begin{equation} \label{eq:b=a}
b(f) = a(Uf) + a^*(Vf) \, ,
\end{equation}
a relation know in the literature as a Bogoliubov transformation \cite{manuceau1968}.
In fact, one has that
\begin{equation} \label{eq:defU+V}
U = \frac{i}{2} \mathcal{F}^{-1} M_{\Gamma_+} \mathcal{F} \quad \mbox{ and } \quad V = \frac{i}{2} \mathcal{F}^{-1} M_{\Gamma_-} \mathcal{F} J
\end{equation}
where $J$ is complex conjugation and $M_{\Gamma_{\pm}}$ is the operator of multiplication
by
\begin{equation} \label{eq:multg}
\Gamma_{\pm}(k) = \frac{1}{\sqrt{ \gamma(k)}} \pm \sqrt{ \gamma(k)} \, ,
\end{equation}
with $\gamma(k)$ as in (\ref{eq:defgamma}).
Using the fact that $\Gamma_{\pm}$ is real valued and even, it is easy to check that
\begin{equation} \label{eq:bog1}
U^* U - V^* V = \idty = U U^* - V V^*
\end{equation}
and
\begin{equation} \label{eq:bog2}
V^* U - U^* V = 0 =  V U^* - UV^* \,
\end{equation}
where we stress that $V^*$ is the adjoint of the {\it anti-linear} mapping $V$. The relation (\ref{eq:b=a}) is invertible, in fact,
\begin{equation}
a(f) = b(U^*f) - b^*(V^*f) \, ,
\end{equation}
and therefore
\begin{equation}
W(f) = \mbox{exp} \left[  \frac{i}{ \sqrt{2}} \left( b((U^*-V^*)f) + b^*((U^*-V^*)f) \right) \right] \, .
\end{equation}
Clearly then,
\begin{equation}
\tau_t(W(f)) = W( T_t f) \, ,
\end{equation}
where the mapping $T_t$ is given by
\begin{equation}
T_t = (U+V) \mathcal{F}^{-1} M_t \mathcal{F} (U^*-V^*) \, ,
\end{equation}
and we have used (\ref{eq:taub}).

\subsubsection{Infinite volume dynamics} It is now clear how to define the infinite volume harmonic dynamics.
Consider a subspace $\mathcal{D} \subset \ell^2( \mathbb{Z}^d)$ and define
$\mathcal{W}( \mathcal{D})$ as above with $\sigma(f,g) = \mbox{Im}[ \langle f, g \rangle ]$.
First assume $\omega >0$, take $\gamma : [- \pi, \pi)^d \to \mathbb{R}$ as in 
(\ref{eq:defgamma}), and
set $U$ and $V$ as in (\ref{eq:defU+V}) with (\ref{eq:multg}). If $\omega >0$, both 
$U$ and $V$ are bounded transformations on $\ell^2( \mathbb{Z}^d)$. We will treat 
the case $\omega =0$ by a limiting argument. The mapping $T_t$ defined by setting
\begin{equation} \label{eq:deftt}
T_t = (U+V) \mathcal{F}^{-1} M_t \mathcal{F} (U^*-V^*) \, ,
\end{equation}
is well-defined on $\ell^2( \mathbb{Z}^d)$. To define the dynamics on $\mathcal{W}(\mathcal{D})$
we will need to choose subspaces $\mathcal{D}$ that are $T_t$ invariant. 
On such $\mathcal{D}$, $T_t$ is clearly real-linear. With (\ref{eq:bog1}) 
and (\ref{eq:bog2}), one
can easily verify the group properties $T_0 = \idty$, $T_{s+t} = T_s \circ T_t$, and
\begin{equation}
\mbox{Im} \left[ \langle T_t f, T_t g \rangle \right] = \mbox{Im} \left[ \langle f,g \rangle \right] \, ,
\end{equation}
i.e. $T_t$ is sympletic in the sense of (\ref{eq:sympT}). Using Theorem 5.2.8 of \cite{bratteli1997}, 
there is a unique 
one parameter group of $*$-automorphisms on $\mathcal{W}( \mathcal{D})$, which we will denote
by $\tau_t$, that satisfies
\begin{equation}
\tau_t(W(f)) = W(T_tf) \quad \mbox{for all } f \in \mathcal{D} \, .
\end{equation}
This defines the harmonic dynamics on $\mathcal{W}( \mathcal{D})$.

Here it is important that $T_t : \mathcal{D} \to \mathcal{D}$. As we demonstrated in \cite{nachtergaele2009},
the mapping $T_t$ can be expressed as a convolution. In fact, 
\begin{equation} \label{eq:defft}
T_tf = f * \overline{ \left(H_t^{(0)} + \frac{i}{2}(H_t^{(-1)} + H_t^{(1)}) \right)} + \overline{f}*\left( \frac{i}{2}(H_t^{(1)} - H_t^{(-1)}) \right).
\end{equation}
where
\begin{equation}\label{eq:h}
\begin{split}
H^{(-1)}_t(x) &= \frac{1}{(2 \pi)^d} {\rm Im} \left[ \int \frac{1}{ \gamma(k)} e^{i(k \cdot x-2\gamma(k)t)} \, d k \right],
\\
H^{(0)}_t(x) &= \frac{1}{(2 \pi)^d}  {\rm Re} \left[ \int e^{i(k \cdot x - 2\gamma(k)t)} \, dk \right],
\\
H^{(1)}_t(x) &=  \frac{1}{(2 \pi)^d} {\rm Im} \left[ \int \gamma(k) \, e^{i(k \cdot x-2\gamma(k)t)} \, dk \right] .
\end{split}
\end{equation}
Using analysis similar to what is proven in \cite{nachtergaele2009}, the following result holds.
\begin{lem}\label{lem:htx}
Consider the functions defined in (\ref{eq:h}). For $\omega\geq 0, \lambda_1,
\ldots,\lambda_d\geq 0$, but such that  
$c_{\omega,\lambda} = (\omega^2 + 4 \sum_{j=1}^d \lambda_j )^{1/2} >0$, 
and any $\mu >0$, the bounds
\begin{equation}
\begin{split}
\left| H_t^{(0)}(x) \right| &\leq  e^{-\mu \left( |x| - c_{\omega,\lambda} \max \left( \frac{2}{\mu} \, , \, e^{(\mu/2)+1}\right) |t| \right)}
\\
\left| H_t^{(-1)}(x) \right| &\le  c^{-1}_{\omega,\lambda}e^{-\mu \left( |x| - c_{\omega,\lambda} \max \left( \frac{2}{\mu} \, , \, e^{(\mu/2)+1}\right) |t| \right)}
\\
\left| H_t^{(1)}(x) \right| &\le c_{\omega,\lambda}e^{\mu/2}e^{-\mu \left( |x| - c_{\omega,\lambda} \max \left( \frac{2}{\mu} \, , \, e^{(\mu/2)+1}\right) |t| \right)}
\end{split}
\end{equation}
hold for all $t \in \mathbb{R}$ and $x \in \mathbb{Z}^d$. Here  $|x| = \sum_{j=1}^{d} |x_i|$.
\end{lem}

Given the estimates in Lemma~\ref{lem:htx}, equation (\ref{eq:defft}) and 
Young's inequality imply that $T_t$ can be defined as
a transformation of $\ell^p(\mathbb{Z}^d)$, for $p\geq 1$. However,
the symplectic form limits us to consider $\mathcal{D}=\ell^p(\mathbb{Z}^d)$
with $1\leq p\leq 2$.

The following bound now readily follows:
\begin{equation}
\begin{split}
\vert \mbox{Im} \langle T_t f, g\rangle\vert \leq & \left(1+2 e^{\mu/2}c_{\omega,\lambda} + 2 c_{\omega,\lambda}^{-1}\right) \times \\
& \quad \times \sum_{x,y} \vert f(x)\vert\, \vert g(y)\vert 
e^{-\mu \left( |x| - c_{\omega,\lambda} \max \left( \frac{2}{\mu} \, , \, e^{(\mu/2)+1}\right) |t| \right)}
\end{split}
\end{equation}
This implies an estimate of the form (\ref{eq:dynestex}), and hence a Lieb-Robinson bound as in (\ref{eq:freest}).

A simple corollary of Lemma~\ref{lem:htx} follows.

\begin{cor} \label{cor:hest} Consider the functions defined in (\ref{eq:h}). For $\omega\geq 0, \lambda_1,
\ldots,\lambda_d\geq 0$, but with 
$c_{\omega,\lambda} = (\omega^2 + 4 \sum_{j=1}^d \lambda_j )^{1/2} >0$.
Take $\| \cdot \|_1$ to be the $\ell^1$-norm. One has that 
\begin{equation}
\| H_t^{(0)} - \delta_0 \|_1 \to 0  \quad \mbox{as} \quad t \to 0,
\end{equation}
and
\begin{equation}
\| H_t^{(m)}  \|_1 \to 0  \quad \mbox{as} \quad t \to 0, \quad \mbox{for } m \in \{ -1, 1\}.
\end{equation}
\end{cor} 
\begin{proof}
The estimates in Lemma~\ref{lem:htx} imply that the functions $H_t^{(m)}$ are bounded by exponentially decaying functions (in $|x|$).
These estimates are uniform for $t$ in compact sets, e.g. $t \in [-1,1]$, and therefore dominated convergence applies. It is clear that
$H_0^{(0)}(x) = \delta_0(x)$ while $H_0^{(m)}(x) = 0$ for $m \in \{-1,1\}$. This proves the corollary.
\end{proof}

\subsubsection{Representing the dynamics} The infinite-volume ground state of the model (\ref{eq:harham}) is the vacuum
state for the $b-$operators, as can be seen from (\ref{eq:diagham}). This state
is defined on $\mathcal{W}(\mathcal{D})$ by
\begin{equation}\label{eq:state}
\rho(W(f))=e^{-\frac{1}{4}\Vert (U^*-V^*)f\Vert^2}
\end{equation}
By standard arguments this defines a state on $\mathcal{W}(\mathcal{D})$
\cite{bratteli1997}. Using (\ref{eq:deftt}), (\ref{eq:bog1}) and (\ref{eq:bog2}) one 
readily verifies that $\rho$ is $\tau_t$-invariant. $\rho$ is regular
by observation. The weak continuity of the dynamics in the GNS-representation
of $\rho$ will follow from the continuity of the functions of the form
\begin{equation}\label{eq:weakcont}
t\mapsto \rho(W(g_1)W(T_t f)W(g_2)), \mbox{ for}\ g_1, g_2, f\in \mathcal{D}.
\end{equation}
When $\omega>0$, this continuity can be easily observed from the
following expresion:
\begin{equation} \label{eq:defstate}
\begin{split}
\rho(W(g_1)W(T_t f)W(g_2)) = & e^{i\sigma(g_1, g_2)/2}
e^{i\sigma(T_t f, g_2 - g_1)/2} \times \\
&\times  e^{-\Vert(U^*-V^*)(g_1+g_2+T_t f)\Vert^2/4}
\end{split}
\end{equation}
Note that $T_t$ is differentiable with bounded derivative and that both $U$ and $V$ 
are bounded. This establishes the continuity in the case that $\omega>0$. 

As discussed in the introduction of the section, the $W^*$-dynamical system
is now defined by considering the GNS representation $\pi_\rho$
of $\rho$. This yields a von Neumann algebra 
$\mathcal{M}=\overline{\pi_\rho(\mathcal{W}(\mathcal{D}))}$. The invariance of
$\rho$ implies that the dynamics is implementable by unitaries $U_t$, i.e.,
\begin{equation}
\pi_\rho(\tau_t(W(f)))=U_t^* \pi_\rho(W(f)) U_t\, .
\end{equation}
Using $U_t$, the dynamics can be extended to $\mathcal{M}$. As a consequence of
(\ref{eq:weakcont}), this extended dynamics is weakly continuous.

\subsubsection{The case of $\omega=0$} We now discuss the case $\omega =0$. Here, the maps $T_t$ are defined
using the convolution formula (\ref{eq:defft}). By Lemma \ref{lem:htx},
$T_t$ is well-defined as a transformation of $\ell^p(\mathbb{Z}^d)$, for
$1\leq p\leq 2$. Both the group property of $T_t$ and the invariance of 
the symplectic form $\sigma$ follow in the limit $\omega\to 0$ by
dominated congervence which is justified by Lemma \ref{lem:htx}.
This demonstrates that the dynamics is well defined.

We represent the dynamics in a state $\rho$ is defined by (\ref{eq:state}), but with the understanding
that $\Vert (U^*- V^*)f\Vert$ may take on the value $+\infty$, in which case
$\rho(W(f))=0$. $\rho$ is still clearly regular. It remains to show that the
dynamics is weakly continuous.

Observe that
\begin{equation}
\begin{split}
T_t f - f =  f * \left(H_t^{(0)} - \delta_0 \right) & - f* \left( \frac{i}{2}(H_t^{(-1)} + H_t^{(1)}) \right) \\
& \quad + \overline{f}*\left( \frac{i}{2}(H_t^{(1)} - H_t^{(-1)}) \right),
\end{split}
\end{equation}
follows from (\ref{eq:defft}). Using Young's inequality and Corollary~\ref{cor:hest}, it is clear that $\| T_t f - f \| \to 0$ as 
$t \to 0$ for any $f \in \ell^p( \mathbb{Z}^d)$ with $1 \leq p \leq 2$. A calculation shows that
\begin{equation}
(U^*-V^*)(T_t f - f) = F_1 *  \left(H_t^{(0)} - \delta_0 \right) - F_2 * H_t^{(-1)}  - i F_3 *H_t^{(1)} \, ,
\end{equation}
where
\begin{equation}
\begin{split}
& \quad \quad F_1 =   \, \mathcal{F}^{-1} M_{\sqrt{ \gamma}} \mathcal{F} \mbox{Im}[f] 
 - i \mathcal{F}^{-1} M_{ \gamma^{-1/2}} \mathcal{F} \mbox{Re}[f] \, , \\
 & F_2 =  \mathcal{F}^{-1} M_{\sqrt{ \gamma}} \mathcal{F} \mbox{Re}[f] \, , \quad \mbox{and} \quad
 F_3 = \mathcal{F}^{-1} M_{\gamma^{-1/2}} \mathcal{F} \mbox{Im}[f] \, . 
\end{split}
\end{equation}
A similar argument to what is given above now implies that 
$\| (U^*-V^*)(T_tf-f) \| \to 0$ as $t \to 0$, for any $f \in \mathcal{D}_0$, 
where
\begin{equation}
\mathcal{D}_0 = \left\{ f \in \ell^2( \mathbb{Z}^d) : \mathcal{F}^{-1} M_{\gamma^{-1/2}} \mathcal{F} \mbox{Re}[f] \in \ell^2( \mathbb{Z}^d) \right\} \, .
\end{equation}
No additional assumption on $\mbox{Im}[f]$ is necessary since $F_3$ is
convolved with $H_t^{(1)}$. Given the form of (\ref{eq:defstate}), this suffices to prove weak continuity.
In fact, one can check that $T_t$ leaves $\mathcal{D}_0$ invariant and that if $f \in \mathcal{D}_0$, then  
$(U^*-V^*)T_t f \in \ell^2( \mathbb{Z}^d)$ for all $t \in \mathbb{R}$. 
This establishes weak continuity of the dynamics, defined on $\mathcal{W}( \mathcal{D}_0)$.

\begin{rem}  We observe that, when $\omega=0$, the finite volume Hamiltonian $H_L^h$ (\ref{eq:harham}) is translation invariant and commutes with the total momentum operator $P_0$ (see (\ref{eq:Q+Pk})). In fact, $H_L^h$ can be written as 
\[ \begin{split} H_L^h &= P_0^2 + \sum_{k \in \Lambda_L^* \backslash \{ 0 \}} P_k^* P_k + \gamma^2 (k) Q_k^* Q_k \\ &= P_0^2 + \sum_{k \in \Lambda_L^* \backslash \{ 0 \}} \gamma(k) (2 b_k^* b_k + 1) \end{split} \] where we used the notation (\ref{eq:Q+Pk}) and, for $k \not = 0$, we introduced the operators $b_k, b_k^*$ as in (\ref{eq:beqns}). 
In this case, the operator $H_L^h$ does not have eigenvectors; its spectrum is purely continuous. By a unitary transformation, the Hilbert space $\cH_{\Lambda_L}$ (see (\ref{eq:hspace})) can be mapped into the space $L^2 (\bR, \rd P_0 ; \cH_b)$ of square integrable functions of $P_0 \in \bR$, with values in $\cH_b$. Here, $\cH_b$ denotes the Fock space generated by all creation and annihilation operators $b_k^*, b_k$ with $k \not = 0$. It is then easy to construct vectors which minimize the energy by a given distribution of the total momentum; for an arbitrary (complex valued) $f \in L^2 (\bR)$ with $\| f \|=1$, we define $\psi_f \in L^2 (\bR, \rd P_0 ; \cH_b)$ by setting $\psi_f (P_0) = f (P_0) \Omega$ (where $\Omega$ is the Fock vacuum in $\cH_b$). These vectors are not invariant with respect to the time evolution. It is simple to check that the Schr\"odinger evolution of $\psi_f$ is given by $e^{-iH_{L}^h t} \psi_f = \psi_{f_t}$ with $f_t (P_0) = e^{-it P_0^2} f (P_0)$ is the free evolution of $f$. In particular, for $\omega =0$, $H_L^h$ does not have a ground state in the traditional sense of an eigenvector. 
For this reason, when $\omega =0$, it is not a priori clear what the natural choice of state should be. As is discussed above, one possibility is to consider first $\omega \not = 0$ and then take the limit $\omega \to 0$. This yields a ground state for the infinite system 
with vanishing center of mass momentum of the oscillators. By considering non-zero values
for the center of mass momentum, one can also define other states with similar properties.
\end{rem}

\subsubsection{Some final comments} The analysis in the following sections and our main result is not limited
to the class of examples we discussed above. E.g., harmonic systems defined on 
more general graphs, such as the ones considered in \cite{eisert2005,eisert2008} 
can also be treated. Also note that our choice of time-invariant
state, while natural, is by no means the only possible. Instead of the
vacuum state defined in (\ref{eq:state}), equilibrium states at positive temperatures could be used in exactly the same way. It would also make
sense to study the convergence of the equilibrium or ground states
for the perturbed dynamics and to consider the dynamics in the representation of the limiting infinite-system state, but we have
not studied this situation and will not discuss it in this paper.

%
%
%
%
%

\section{Perturbing the Harmonic Dynamics}

In this section, we will discuss finite volume perturbations of the infinite volume 
harmonic dynamics which we defined in Section~\ref{sec:harm}. To begin, we recall a
fundamental result about perturbations of quantum dynamics defined by adding a bounded term
to the generator. This is a version of what is usually known as the Dyson or Duhamel 
expansion. The following statement summarizes Proposition 5.4.1 of \cite{bratteli1997}.
\begin{prop} \label{prop:perdyn} Let $\{ \mathcal{M}, \alpha_t \}$ be a $W^*$-dynamical system and let $\delta$ denote the
infinitesimal generator of $\alpha_t$. Given any $P = P^* \in \mathcal{M}$, set $\delta_P$ to be the bounded 
derivation with domain $D( \delta_P) = \mathcal{M}$ satisfying $\delta_P(A) = i [P, A]$ for all $A \in \mathcal{M}$.
It follows that $\delta + \delta_P$ generates a one-parameter group of $*$-automorphisms $\alpha^P$ of $\mathcal{M}$
which is the unique solution of the integral equation
\begin{equation} \label{eq:defpdyn}
\alpha_t^P(A) = \alpha_t(A) + i \int_0^t \alpha_s^P \left( \left[ P, \alpha_{t-s}(A) \right] \right) \, ds \, . 
\end{equation}
In addition, the estimate
\begin{equation} \label{eq:dynestpert}
\left\| \alpha_t^P(A) - \alpha_t(A)  \right\| \leq \left( e^{|t| \| P \|} - 1 \right) \, \| A \| \, 
\end{equation}
holds for all $t \in \mathbb{R}$ and $A \in \mathcal{M}$.
\end{prop}

Since the initial dynamics $\alpha_t$ is assumed weakly continuous, the norm estimate (\ref{eq:dynestpert}) can be used to
show that the perturbed dynamics is also weakly continuous. Hence, for each $P = P^* \in \mathcal{M}$ the pair 
$\{ \mathcal{M}, \alpha_t^P \}$ is also a $W^*$-dynamical system. Thus, if $P_i=P_i^* \in \mathcal{M}$ 
for $i=1,2$, then one can define $\alpha_t^{P_1+P_2}$ iteratively.


\subsection{A Lieb-Robinson bound for on-site perturbations} \label{subsec:onsite}

In this section we will consider perturbations of the harmonic dynamics 
defined in Section~\ref{sec:harm}. Recall that our general assumptions for the harmonic
dynamics on $\Gamma$ are as follows.

We assume that the harmonic dynamics, $\tau^0_t$, is defined on a Weyl algebra
$\mathcal{W}( \mathcal{D})$ where $\mathcal{D}$ is a subspace of $\ell^2(\Gamma)$.
In fact, we assume there exists a group $T_t$ of real-linear
transformations which leave $\mathcal{D}$ invariant and satisfy
\begin{equation}
\tau_t^0(W(f)) = W(T_tf) \quad \mbox{for all } f \in \mathcal{D} \, .
\end{equation}
In addition, we assume that this harmonic dynamics satisfies a Lieb-Robinson 
bound. Specifically, we suppose that there exists a number $a_0 >0$ for which
given any $0<a \leq a_0$, there are positive numbers $c_a$ and $v_a$ for which
\begin{equation} \label{eq:prelrb}
\left| 1 - e^{ i \sigma(T_tf, g)} \right| \leq c_a e^{v_a|t|} \sum_{x,y \in \Gamma} |f(x)| \, |g(y)| \, F_a \left( d(x,y) \right) \,
\end{equation}
here the spatial decay in $\Gamma$ is described by the function $F_a$ as introduced in
Section~\ref{sec:bdints}. As we discussed in Section~\ref{sec:harm}, the estimate (\ref{eq:prelrb})
immediately implies the Lieb-Robinson bound
\begin{equation} \label{eq:freelrb}
\left\| \left[ \tau_t^0 \left( W(f) \right), W(g) \right] \right\| \leq c_a e^{v_a|t|} \sum_{x,y \in \Gamma} |f(x)| \, |g(y)| \, F_a \left( d(x,y) \right) \, .
\end{equation}
Finally, we assume that we have represented this harmonic dynamics in a regular and
$\tau_t^0$-invariant state $\rho$ for which the pair $\{ \mathcal{M}, \tau_t^0 \}$, with
$\mathcal{M} = \overline{\pi_{\rho}(\mathcal{W}( \mathcal{D}))}$, is a $W^*$-dynamical system.

Our first estimate involves perturbations defined as finite sums of on-site terms.
More specifically, the perturbations we consider are defined as follows.

To each site $x \in \Gamma$, we will associate a finite measure $\mu_x$ on $\mathbb{C}$, and
an element $P_x \in \mathcal{W}(\mathcal{D})$ which has the form
\begin{equation} \label{eq:defpx}
P_x = \int_{\mathbb{C}} W(z \delta_x) \mu_x(dz) \, .
\end{equation}
We require that each $\mu_x$ is even, i.e. invariant under $z \mapsto -z$, to ensure 
self-adjointness, i.e. $P_x^* = P_x$. Our Lieb-Robinson bounds hold 
under the additional assumption that the second moment
is uniformly bounded, i.e. 
\begin{equation} \label{eq:totbd}
\sup_{x \in \Gamma} \int_{\mathbb{C}} |z|^2 \, | \mu_x |( dz ) < \infty \, .
\end{equation}

We use Proposition~\ref{prop:perdyn} to define the perturbed dynamics.
Fix a finite set $\Lambda \subset \Gamma$. Set
\begin{equation} \label{eq:defpl}
P^{\Lambda} = \sum_{x \in \Lambda} P_x \, ,
\end{equation}
and note that $(P^{\Lambda})^* = P^{\Lambda} \in \mathcal{W}(\mathcal{D})$. 
We will denote by $\tau_t^{(\Lambda)}$ the dynamics that results from 
applying Proposition~\ref{prop:perdyn} to the $W^*$-dynamical system $\{ \mathcal{M}, \tau_t^0 \}$
and $P^{\Lambda}$.

Before we begin the proof of our estimate, we discuss two examples.
\begin{ex} 1) Let $\mu_x$ be supported on $[- \pi, \pi)$ and absolutely 
continuous with respect to Lebesgue measure, i.e. 
$\mu_x(dz) = v_x(z) dz$. If $v_x$ is in $L^2([-\pi,\pi))$, then $P_x$ is proportional to an operator of multiplication by the 
inverse Fourier transform of $v_x$. Moreover, since the support of $\mu_x$ is real, $P_x$ corresponds to 
multiplication by a function depending only on $q_x$.
\end{ex}
\begin{ex}
2) Let $\mu_x$ have finite support, e.g., take $\mbox{supp}( \mu_x) = \{ z, -z \}$ for some number
$z = \alpha +  i \beta \in \mathbb{C}$. Then
\begin{equation}
P_x = W( z \delta_x) + W(- z \delta_x) = 2 \cos( \alpha q_x + \beta p_x) \, .
\end{equation}
\end{ex}

We now state our first result.

\begin{thm} \label{thm:ahlrb}
Let $\tau_t^0$ be a harmonic dynamics defined on $\Gamma$ as described above. Suppose that
\begin{equation}
\kappa = \sup_{x \in \Gamma} \int_{\mathbb{C}} |z|^2 | \mu_x| (d z) < \infty \, ,
\end{equation}
and define the perturbed dynamics $\tau_t^{(\Lambda)}$ as indicated above.
For every $0<a \leq a_0$, there exist positive numbers $c_a$ and $v_a$ for which
the estimate
\begin{equation} \label{eq:anharmbd}
\left\| \left[ \tau_t^{(\Lambda)} \left( W(f) \right), W(g) \right] \right\| \leq c_a e^{ (v_a + c_a \kappa C_a ) |t|} \sum_{x, y} |f(x)| \, |g(y)| F_a \left( d(x,y) \right)
\end{equation}
holds for all $t \in \mathbb{R}$ and for any functions $f, g \in \mathcal{D}$. 
\end{thm}

Here the numbers $c_a$ and $v_a$ are as in (\ref{eq:prelrb}), whereas $C_a$ is the convolution constant
as defined in (\ref{eq:intlat}) with respect to the function $F_a$.
\begin{proof}
Fix $t >0$ and define the function $\Psi_t : [0,t] \to \mathcal{W}(\mathcal{D})$ by setting
\begin{equation} \label{eq:defpsit}
\Psi_t(s) = \left[  \tau_s^{(\Lambda)} \left( \tau_{t-s}^0(W(f)) \right), W(g) \right] \, .
\end{equation}
It is clear that $\Psi_t$ interpolates between the commutator associated with the original harmonic dynamics, $\tau_t^0$ at $s=0$, 
and that of the perturbed dynamics, $\tau_t^{(\Lambda)}$ at $s=t$. A calculation shows that
\begin{equation} \label{eq:dpsit}
\frac{d}{ds} \Psi_t(s) = i \sum_{x \in \Lambda} \left[ \, \tau_s^{(\Lambda)} \left( \left[ P_x, W ( T_{t-s}f) \right] \right) , W(g) \right] \, ,
\end{equation}
where differentiability is guaranteed by the results of Proposition~\ref{prop:perdyn}.
The inner commutator can be expressed as
\begin{eqnarray}
 \left[ P_x, W ( T_{t-s}f) \right] & = & \int_{\mathbb{C}} \left[ W(z \delta_x), W( T_{t-s}f) \right] \mu_x( dz) \nonumber \\
 & = & W( T_{t-s}f) \mathcal{L}_{t-s;x}(f)  \, .
\end{eqnarray}
where
\begin{equation} \label{eq:defLx}
\mathcal{L}^*_{t-s;x}(f) = \mathcal{L}_{t-s;x}(f) =  \int_{\mathbb{C}} W(z \delta_x) \left\{e^{i \sigma( T_{t-s}f, z \delta_x)} -1 \right\} \mu_x(dz) \, \in \mathcal{W}(\mathcal{D}) \, .
\end{equation}
Thus $\Psi_t$ satisfies
\begin{equation}
\begin{split}
\frac{d}{ds} \Psi_t(s) =   i \sum_{x \in \Lambda} & \Psi_t(s)  \tau_s^{(\Lambda)} \left( \mathcal{L}_{t-s;x}(f) \right)  \\ 
+ & i \sum_{x \in \Lambda} \tau_s^{(\Lambda)} \left( W( T_{t-s}f) \right) \, \left[ \tau_s^{(\Lambda)} \left( \mathcal{L}_{t-s;x}(f) \right), W(g) \right]  \, .
\end{split}
\end{equation}
The first term above is norm preserving. In fact, define a unitary evolution $U_t(\cdot)$ by setting
\begin{equation}
\frac{d}{ds}U_t(s) = - i \sum_{x \in \Lambda} \tau_s^{(\Lambda)} \left( \mathcal{L}_{t-s;x}(f) \right) U_t(s) \quad \mbox{with } U_t(0) = \idty \, .
\end{equation}
It is easy to see that
\begin{equation}
\frac{d}{ds} \left( \Psi_t(s) U_t(s) \right) =  i \sum_{x \in \Lambda} \tau_s^{(\Lambda)} \left( W( T_{t-s}f) \right) \, \left[ \tau_s^{(\Lambda)} \left( \mathcal{L}_{t-s;x}(f) \right), W(g) \right] U_t(s) \, ,
\end{equation}
and therefore,
\begin{equation}
\Psi_t(t) U_t(t) = \Psi_t(0) + i \sum_{x \in \Lambda} \int_0^t  \tau_s^{(\Lambda)} \left( W( T_{t-s}f) \right) \, \left[ \tau_s^{(\Lambda)} \left( \mathcal{L}_{t-s;x}(f) \right), W(g) \right] U_t(s) \, ds \, .
\end{equation}
Estimating in norm, we find that
\begin{equation} \label{eq:norm1}
\begin{split}
\Big\| \Big[ \tau_t^{(\Lambda)} \left(W (f)\right)  , W(g) \Big] \Big\|  \leq &  \Big\| \Big[ \tau_t^0 \left(W (f)\right) , W(g) \Big] \Big\| \\
& +  \sum_{x \in \Lambda} \int_0^t  \Big\| \left[  \tau_s^{(\Lambda)} \left( \mathcal{L}_{t-s;x}(f) \right) , W(g) \right] \Big\| \, ds \, .
\end{split}
\end{equation}
Moreover, using (\ref{eq:defLx}) and the bound (\ref{eq:prelrb}), it is clear that
\begin{equation} \label{eq:norm2}
\begin{split}
 \Big\| \left[  \tau_s^{(\Lambda)} \left( \mathcal{L}_{t-s;x}(f) \right) , W(g) \right] \Big\| \leq & \, c_a e^{v_a(t-s)} \sum_{x' \in \Gamma} |f(x')| F_a \left( d(x,x') \right) \times \\
& \quad \times \int_{\mathbb{C}} |z|  \, \Big\| \left[  \tau_s^{(\Lambda)} \left( W(z \delta_x) \right) , W(g) \right] \Big\| \, |\mu_x|(dz)
\end{split}
\end{equation}
holds. Combining (\ref{eq:norm2}), (\ref{eq:norm1}), and (\ref{eq:freelrb}), we have proven that
\begin{equation}\label{eq:norm3}
\begin{split}
\Big\| \Big[ \tau_t^{(\Lambda)} \left(W (f)\right)  , W(g) \Big] \Big\| \leq \; & c_a e^{v_a t} \sum_{x, y} |f(x)| \, |g(y)| \, F_a \left( d(x,y) \right) \\ 
&+ c_a \sum_{x' \in \Gamma} |f(x')| \sum_{x \in \Lambda} F_a \left( d(x,x') \right) \int_0^t  e^{v_a (t-s)}  \times \\
&\quad  \times \int_{\mathbb{C}} |z| \, \Big\| \left[  \tau_s^{(\Lambda)} \left( W(z \delta_x) \right) , W(g) \right] \Big\| \, |\mu_x|(dz) \, ds \,.
\end{split}
\end{equation}
Following the iteration scheme applied in \cite{nachtergaele2009}, one arrives at (\ref{eq:anharmbd}) as claimed.
\end{proof}


\subsection{Multiple Site Anharmonicities} \label{subsec:anharmms}

In this section, we will prove that Lieb-Robinson bounds, similar to those in
Theorem~\ref{thm:ahlrb}, also hold for perturbations involving short range interations.
We introduce these as follows.

For each finite subset $X \subset \Gamma$, we associate a finite measure $\mu_X$ on
$\mathbb{C}^X$ and an element $P_X \in \mathcal{W}(\mathcal{D})$ with the form
\begin{equation} \label{eq:PX}
P_X = \int_{\mathbb{C}^X} W( z \cdot \delta_X) \, \mu_X(d z) \, ,
\end{equation}
where, for each $z \in \mathbb{C}^X$, the function $z \cdot \delta_X : \Gamma \to \mathbb{C}$ is
given by
\begin{equation}
(z \cdot \delta_X)(x)  = \sum_{x' \in X} z_{x'} \delta_{x'}(x) = \left\{ \begin{array}{cc} z_x & \mbox{if } x \in X , \\ 0 & \mbox{otherwise.} \end{array} \right.
\end{equation}
We will again require that $\mu_X$ is invariant with respect to $z \mapsto -z$, and hence, $P_X$ is
self-adjoint. In analogy to (\ref{eq:defpl}), for any finite subset $\Lambda \subset \Gamma$, we will set
\begin{equation} \label{eq:defpl2}
P^{\Lambda} = \sum_{X \subset \Lambda} P_X\, ,
\end{equation}
where the sum is over all subsets of $\Lambda$.
Here we will again let $\tau^{(\Lambda)}_t$ denote the dynamics resulting from Proposition~\ref{prop:perdyn}
applied to the $W^*$-dynamical system $\{ \mathcal{M}, \tau_t^0 \}$ and the perturbation $P^{\Lambda}$ defined by (\ref{eq:defpl2}). 

The main assumption on these multi-site perturbations follows.
There exists a number $a_1 >0$ such that for all $0< a \leq a_1$, there is a number $\kappa_a >0$ for which
given any pair $x_1, x_2 \in \Gamma$, 
\begin{equation} \label{eq:pertbd}
\sum_{\stackrel{X \subset \Gamma:}{x_1, x_2 \in X}} \int_{\mathbb{C}^X} |z_{x_1}| | z_{x_2}| \big| \mu_X \big|(dz) \leq \kappa_a F_a \left( d(x_1,x_2) \right) \, .
\end{equation} 

\begin{thm} \label{thm:ahlrbms} Let $\tau_t^0$ be a harmonic dynamics defined on $\Gamma$.
Assmue that (\ref{eq:pertbd}) holds, and that $\tau_t^{(\Lambda)}$ denotes the corresponding
perturbed dynamics. For every $0< a \leq \min(a_0, a_1)$, there exist positive numbers $c_a$ and $v_a$ for which
the estimate
\begin{equation} \label{eq:anharmbdms}
\left\| \left[ \tau_t^{(\Lambda)} \left( W(f) \right), W(g) \right] \right\| \leq c_a e^{ (v_a +  c_a \kappa_a C_a^2 ) |t|} \sum_{x, y} |f(x)| \, |g(y)| F_a \left( d(x,y) \right)
\end{equation}
holds for all $t \in \mathbb{R}$ and for any functions $f, g \in \mathcal{D}$. 
\end{thm}
The proof of this result closely follows that of Theorem~\ref{thm:ahlrb}, and so we only comment on
the differences.
\begin{proof}
For $f,g \in \mathcal{D}$ and $t >0$, define $\Psi_t :[0,t] \to \mathcal{W}(\mathcal{D})$ as in (\ref{eq:defpsit}).
The derivative calculation beginning with (\ref{eq:dpsit}) proceeds as before. Here
\begin{equation} \label{eq:deflzms}
\mathcal{L}_{t-s;X}(f) = \int_{\mathbb{C}^X} W( z \cdot \delta_X) \left\{ e^{i \sigma(T_{t-s}f, z \cdot \delta_X) } - 1 \right\} \, \mu_X( d z) \, ,
\end{equation}
is also self-adjoint. The norm estimate
\begin{equation} \label{eq:norm1ms}
\begin{split}
\Big\| \Big[ \tau_t^{(\Lambda)} \left(W (f)\right)  , W(g) \Big] \Big\|  \leq  & \Big\| \Big[ \tau_t^0 \left(W (f)\right) , W(g) \Big] \Big\| \\ & +  \sum_{X \subset \Lambda} \int_0^t  \Big\| \left[  \tau_s^{(\Lambda)} \left( \mathcal{L}_{t-s;X}(f) \right) , W(g) \right] \Big\| \, ds \, ,
\end{split}
\end{equation}
holds similarly. With (\ref{eq:deflzms}), it is easy to see that the integrand in (\ref{eq:norm1ms}) is bounded by
\begin{equation} 
c_a e^{v_a(t-s)} \sum_{x\in \Gamma} \, |f(x)| \, \sum_{x' \in X} F_a \left( d(x,x') \right) \int_{\mathbb{C}^X}  | z_{x'}|  \,  \, \Big\| \left[  \tau_s^{(\Lambda)} \left( W(z \cdot \delta_X) \right) , W(g) \right] \Big\| \, |\mu_X|(dz) \, ,
\end{equation}
the analogue of (\ref{eq:norm2}), for $0<a \leq a_0$. Moreover, if $0<a \leq \min(a_0,a_1)$, then 
\begin{equation}\label{eq:norm4}
\begin{split}
\Big\| \Big[ \tau_t^{(\Lambda)} \left(W (f)\right)  , W(g) \Big] \Big\| \leq \; &  c_a e^{v_a t} \sum_{x, y \in \Gamma} |f(x)| \, |g(y)| \, F_a \left( d(x,y) \right) \\ &+ c_a \sum_{x \in \Gamma} |f(x)| \sum_{X \subset \Lambda} \sum_{x' \in X} F_a \left( d(x,x') \right) \times \\
& \quad \times \int_0^t  e^{v_a (t-s)}  
\int_{\mathbb{C}^X} |z_{x'}| \, \Big\| \left[  \tau_s^{(\Lambda)} \left( W(z \cdot \delta_X) \right) , W(g) \right] \Big\| \, |\mu_X|(dz) \, ds \, .
\end{split}
\end{equation}
The estimate claimed in (\ref{eq:anharmbdms}) follows by iteration. In fact, the first term in the iteration is bounded by
\begin{equation}
\begin{split}
c_a \sum_x |f(x)| \sum_{X \subset \Lambda} \sum_{x_1 \in X} & F_a \left( d(x,x_1) \right) \int_0^t  e^{v_a (t-s)}  \\  
& \times \int_{\mathbb{C}^X} |z_{x_1}| \, \Big( c_a e^{v_a s} \sum_{x_2 \in X} \sum_y |z_{x_2}| \, |g(y)| \, F_a \left( d(x_2,y) \right) \Big) \, |\mu_X|(dz) \, ds \, \\
\leq  c_a  t \cdot c_a e^{v_a t}  \sum_{x,y} & |f(x)| |g(y)|  \sum_{x_1, x_2 \in \Gamma} F_a \left( d(x,x_1) \right)  F_a \left( d(x_2,y) \right) \sum_{\stackrel{X \subset \Gamma:}{x_1,x_2 \in X}} \int_{\mathbb{C}^X} | z_{x_1}| |z_{x_2}| | \mu_{X}|(dz) \, \\
\leq   \kappa_a c_a t  \cdot  c_a e^{v_a t} \sum_{x,y} & |f(x)| |g(y)|  \sum_{x_1, x_2 \in \Gamma} F_a \left( d(x,x_1) \right) F_{a} \left( d(x_1, x_2) \right)  F_a \left( d(x_2,y) \right) \\
\leq  \kappa_a C_{a}^2 c_{a} t  \cdot  c_a e^{v_a t} \sum_{x,y} & |f(x)| |g(y)| F_a \left( d(x,y) \right) \, .
\end{split}
\end{equation}
The higher order iterates are treated similarly.
\end{proof}

%
%
%

\section{Existence of the Dynamics}
In this section, we demonstrate that the finite volume dynamics analyzed in the previous section
converge to a limiting dynamics as the volume $\Lambda$ on which the perturbation is defined
tends to $\Gamma$. We state this as Theorem~\ref{thm:exist} below.

\begin{thm} \label{thm:exist}
Let $\tau_t^0$ be a harmonic dynamics defined on $\mathcal{W}( \ell^1( \Gamma) )$ as described in Section~\ref{subsec:onsite}.
Let $\{ \Lambda_n \}$ denote a non-decreasing, exhaustive sequence of finite subsets of $\Gamma$. 
Consider a family of perturbations $P^{\Lambda_n}$ as defined in (\ref{eq:defpl2}) and (\ref{eq:PX}) which satisfy
(\ref{eq:pertbd}). Suppose in addition that
\begin{equation} \label{eq:1mom}
M = \sup_{x \in \Gamma} \sum_{\stackrel{X \subset \Gamma:}{x \in X}} \int_{\mathbb{C}^X} |z_x| | \mu_X|(dz) < \infty \, .
\end{equation}
Then, for each $f \in \ell^1( \Gamma)$ and $t \in \mathbb{R}$ fixed, the limit
\begin{equation} \label{eq:lim}
\lim_{n \to \infty} \tau_t^{(\Lambda_n)} \left( W(f) \right) \, 
\end{equation}
exists in norm. The limiting dynamics, which we denote by $\tau_t$, is weakly continuous.
\end{thm}

It is important to note that since the estimates in Theorem~\ref{thm:ahlrbms} are independent of $\Lambda$,
the limiting dynamics also satisfies a Lieb-Robinson bound as in (\ref{eq:anharmbdms}). We now
prove Theorem~\ref{thm:exist}.

\begin{proof}
Fix a Weyl operator $W(f)$ with $f \in \ell^1(\Gamma)$. 
Let $T>0$ and take $m \leq n$. Iteratively applying Proposition~\ref{prop:perdyn}, we have that
\begin{equation} \label{eq:dyn3}
\tau_t^{(\Lambda_n)}(W(f)) = \tau_t^{(\Lambda_m)}(W(f)) + i \int_0^t \tau_s^{(\Lambda_n)} \left( \left[ P^{\Lambda_n \setminus \Lambda_m}, \tau_{t-s}^{(\Lambda_m)}(W(f)) \right] \right) \, ds \, ,
\end{equation}
for all $-T \leq t \leq T$. The bound
\begin{eqnarray}
&&\left\| \left[ P^{\Lambda_n \setminus \Lambda_m}, \tau_{t-s}^{(\Lambda_m)}(W(f)) \right] \right\|\\
 & \leq & \sum_{\stackrel{X \subset \Lambda_n:}{ X \cap \Lambda_n \setminus \Lambda_m \neq \emptyset}} \int_{\mathbb{C}^X}
\left\| \left[ W(z \cdot \delta_X) , \tau_{t-s}^{(\Lambda_m)}(W(f)) \right] \right\| \, | \mu_X| ( dz) \nonumber \\
& \leq & c_a e^{(v_a +c_a \kappa_a C_a^2)(t-s)} \sum_{x \in \Gamma} |f(x)| \sum_{\stackrel{X \subset \Lambda_n:}{ X \cap \Lambda_n \setminus \Lambda_m \neq \emptyset}} \sum_{y \in X} F_a \left( d(x,y) \right) \int_{\mathbb{C}^X} |z_y| \, | \mu_X| (dz) \nonumber \\ & \leq & c_a e^{(v_a +c_a \kappa_a C_a^2)(t-s)} \sum_{x \in \Gamma} |f(x)| \sum_{y \in \Lambda_n \setminus \Lambda_m }   F_a \left( d(x,y) \right) \sum_{\stackrel{X \subset \Gamma:}{ y \in X}}   \int_{\mathbb{C}^X} |z_y| \, | \mu_X| (dz) \nonumber \\ & \leq & M c_a e^{(v_a +c_a \kappa_a C_a^2)(t-s)} \sum_{x \in \Gamma} |f(x)| \sum_{y \in \Lambda_n \setminus \Lambda_m }   F_a \left( d(x,y) \right) \nonumber
\end{eqnarray}
follows readily from Theorem~\ref{thm:ahlrbms} and assumption (\ref{eq:1mom}).
For $f \in \ell^1( \Gamma)$ and fixed $t$, the upper estimate above goes to zero as
$n,m \to \infty$. In fact, the convergence is uniform for $t \in [-T,T]$. 
This proves (\ref{eq:lim}).

By an $\epsilon/3$ argument, similar to what is done at the end of Section~\ref{sec:bdints}, weak continuity follows since we know it holds for the 
finite volume dynamics.
This completes the proof of Theorem~\ref{thm:exist}.
\end{proof}

%
%
%
%
%


\subsection*{Acknowledgment}

The work reported on in this paper was supported by the National Science
Foundation: B.N. under Grants \#DMS-0605342 and \#DMS-0757581,
R.S. under Grant \#DMS-0757424, and S.S. under Grant \#DMS-0757327
and \#DMS-0706927. The authors would also like to acknowledge the hospitality 
of the department of mathematics at U.C. Davis where a part of this work was
completed.
\end{document}